\documentclass[11pt]{amsart}
\usepackage{fullpage}
\usepackage{amssymb}
\usepackage{graphicx}
\usepackage{epstopdf}
\usepackage{url}
\usepackage{amsmath}
\usepackage{amsthm}
\usepackage{fullpage}

\newcommand{\prob}[1]{\textsf{#1}}  
\newcommand{\myparagraph}[1]{\smallskip\noindent\textbf{#1}\quad}

\newcommand{\alg}[1]{\textbf{#1}}   

\def\calP{{\mathcal P}}   
\def\calM{{\mathcal M}}   
\def\calU{{\mathcal U}}   
\newcommand{\PCopt}{\overline{OPT}}
\def\affm{\mathbf{A}}
\newcommand{\trace}{{Se}}

\newlength{\tablength}
\newlength{\spacelength}
\newcommand{\tabstar}{\hspace*{\tablength}}
\newcommand{\spacestar}{\hspace*{\spacelength}}
\def\obeytabs{\catcode`\^^I=\active}
{\obeytabs\global\let^^I=\tabstar}
{\obeyspaces\global\let =\spacestar}

\newenvironment{display}{\begingroup\obeylines\obeyspaces\obeytabs}{\endgroup}
\newenvironment{prog}{\begin{display}\parskip0pt\sf}{\end{display}}

\newenvironment{pseudo}{\begin{quote}\begin{prog}}{\end{prog}\end{quote}}

\newcounter{foo}
\newtheorem{theorem}[foo]{Theorem}
\newtheorem{lemma}[foo]{Lemma}
\newtheorem{corollary}[foo]{Corollary}

\newtheorem{defn}[foo]{Definition}

\newtheorem{proposition}{Proposition}[section]
\newtheorem{claim}[foo]{Claim}
\newcommand{\biopt}{OPT^B}
\newcommand{\bipcopt}{\PCopt^B}
\def\hi{{1}}
\def\hj{{\hat{j}}}

\begin{document}
\title[Wireless Capacity in General Metrics]{Wireless Capacity with Oblivious Power in General Metrics}

\author[M. Halld\'orsson]{Magn\'us M. Halld\'orsson}
\address[M. Halld\'orsson]{School of Computer Science\\
Reykjavik University\\
Reykjavik 101, Iceland}
\email{mmh@ru.is}

\author[P. Mitra]{Pradipta Mitra}
\address[P. Mitra]{School of Computer Science\\
Reykjavik University\\
Reykjavik 101, Iceland}
\email{ppmitra@gmail.com}

\begin{abstract} 
\small
 The capacity of a wireless network is the maximum possible amount of
simultaneous communication, taking interference into account. 
Formally, we treat the following problem.  Given is a set of links,
each a sender-receiver pair located in a metric space, and an assignment
of power to the senders.  We seek a maximum subset of links that are
feasible in the SINR model: namely, the signal received on each link
should be larger than the sum of the interferences from the other
links.  We give a constant-factor approximation that holds for any
length-monotone, sub-linear power assignment and any distance metric.

We use this to give essentially tight characterizations of capacity
maximization under power control using oblivious power assignments.
Specifically, we show that the mean power assignment is optimal for
capacity maximization of bi-directional links, and give a tight
$\theta(\log n)$-approximation of scheduling bi-directional links with
power control using oblivious power. For uni-directional links we give
a nearly optimal $O(\log n + \log \log \Delta)$-approximation to the
power control problem using mean power, where $\Delta$ is the ratio of
longest and shortest links. Combined, these results clarify
significantly the centralized complexity of wireless communication
problems.
\end{abstract}

\maketitle

\section{Introduction.}

How much communication can occur simultaneously in any given wireless
network? This basic question addressed here of the \emph{capacity} of
a wireless network has attracted tremendous attention, led by 
the highly-cited work on non-constructive average-case analysis by
Gupta and Kumar \cite{kumar00}. Less is known about algorithmic
results on arbitrary instances.


In the problems we consider, we are given a set of links in a
metric space.  In the \emph{uni-directional} model, the link $\ell_v$
represents a directed communication request from a sender $s_v$ to a
receiver $r_v$, each a point in the space. In the
\emph{bi-directional} model, communication is two way, thus both
points in the link act as sender and receiver.  In either model, a power $P_v$ 
is associated with each
link $\ell_v$  (which may part of the input, or
something the algorithm must find out).  When communicating during the
same time, two links $\ell_u$ and $\ell_v$ interfere with each other.
We adopt the physical model or ``SINR model" of interference,
precisely defined in the Section \ref{sec:model}.
It is known to capture reality more faithfully than the graph-based
models most common in the theory literature, as shown 
theoretically as well as
experimentally~\cite{GronkMibiHoc01,MaheshwariJD08,Moscibroda2006Protocol}.

Within this overall framework, a large number of problems can be defined.
Even limited to the specific scope of this paper, at least the
following dimensions in the problem space can be considered. 

First, one can consider \emph{fixed} versus \emph{arbitrary} power.
In the fixed version, the power that a link must use is a part of the input,
whereas
for the arbitrary power (or, ``power control") case, 
choosing the appropriate powers to optimize capacity is a part of the problem. 
For ease of implementation in distributed settings,
research on both fixed and arbitrary power
has focused on \emph{oblivious} power assignments, where the power
of a link is a (usually simple) function of the length of the link. 
We study how well oblivious power
assignments can handle the power control problem and provide
tight characterization (upper and lower bounds) for certain cases.

Second, there is the \emph{capacity} problem (also known as the
\emph{single slot scheduling} problem), of finding the
maximum number of links that can communicate simultaneously
 vs.\ the \emph{scheduling}
problem, which is to partition the links into the smallest number of 
slots such that each slot can be scheduled at
once. 
We focus mainly on the capacity problem, but study the
relationship between capacity and scheduling, specially in the context
of the oblivious vs.\ arbitrary power dichotomy.

Finally, there is the issue of the metric space.
The common assumption of the 2-dimensional Euclidean plane is 
natural, but for wireless transmission it is clearly a grand simplification.
Antenna directionality, walls and obstructions, 
environmental conditions, and terrains all complicate and distort signal
transmissions, questioning the geometric assumption. 
In this work, we are interested in exploring what can be said in the case of
completely arbitrary metrics. 
%
This relates also to the issue of the \emph{path loss exponent} $\alpha$
 (defined in Section \ref{sec:model}),
a constant that defines how quickly the signal falls off from its source. 
Most approximation results have depended on the 
assumption that $\alpha > 2$
\cite{moscibroda06b,MoscibrodaOW07,chafekar07,GHWW09,HW09}, while it
is known that $\alpha$ can actually be equal to or smaller
than 2 in real networks (see \cite{SWTFA95}, for example). 

The aim of our work is to clarify the relationships between the
different problem variants, as well as obtaining improved bounds for
the fundamental problems.

\subsection{Our Contributions.}


For general metrics, we
provide a $O(1)$-approximation algorithm for the capacity problem for all length-monotone, sub-linear
power assignments (these terms are defined in Section \ref{sec:model}). In particular, this implies constant factor approximation for uniform, mean-power
and linear-power assignments. The results apply for both uni- and bi-directional links. Previously,
an $O(1)$-approximation was known only for uniform power 
in the plane with $\alpha > 2$ \cite{HW09}. 

We use the above results to study the power control problem through the use of oblivious power, specifically, mean power.
For bi-directional links,
we achieve a $O(1)$-approximation algorithm 
improving the previous $O(\log n)$-approximation for metrics with bounded Assouad dimension (or \emph{doubling} metrics)\cite{us:esa09full}, where $n$ is the number of links.
For uni-directional links, we improve and generalize the $O(\log n
\cdot \log\log \Delta)$-approximation of \cite{us:esa09full} for
\emph{doubling} metrics to a $O(\log n + \log \log \Delta)$-approximation, 
where $\Delta$ is the ratio between longest to the shortest linklength.
This nearly matches the lower bound of $\Omega(\log \log \Delta)$ for oblivious power assignments \cite{us:esa09full}.

Algorithmic results for capacity imply approximation algorithms for the corresponding
scheduling problems with an extra multiplicative $O(\log n)$ factor. In all cases,
these improve and/or generalize existing bounds. Among these, the $O(\log n)$-approximation
algorithm for scheduling bi-directional links with power control deserves special mention.
It improves on previous poly-logarithmic bounds \cite{KV10,FKRV09}, and we additionally
show it to be best possible for any algorithm using oblivious power assignment.
This is done by providing an example of links in a metric space that can be scheduled using some
power assignment in one slot, yet require $\Omega(\log n)$
slots using oblivious power. 
Observe that in this construction all links are of equal length, for
which uniform power is known to result in a constant-factor
approximation in the 2-dimensional setting (with $\alpha > 2$) \cite{us:esa09full}. Hence,
this also implies that when scheduling, general metrics are strictly
harder than doubling metrics. 

Another feature of our results is that they work for all positive values of the path loss exponent $\alpha$, generalizing
a host of results that required $\alpha > 2$.
Finally, noting that the approximation factors we achieve are exponential in $\alpha$ (which is a constant
for typical values of $\alpha$), we 
prove an inapproximability result showing that in general,
exponential dependence on the path loss exponent 
is necessary.


The techniques used can be viewed as a step in the evolution of 
SINR analysis. The key addition is to treat both the
interference caused as well as interference received. We treat the
pairwise affectances as a matrix, and build on sparsity properties.
In addition, we extend the key technical lemma -- the \emph{red-blue
  lemma} -- of \cite{GHWW09,HW09} from the plane to arbitrary distance metrics.

\subsection{Related Work.}

Early work on scheduling in the SINR model was focused on heuristics
and/or non-algorithmic average-case analysis (e.g.~\cite{kumar00}). 
In contrast, the body of algorithmic work is mostly on graph-based models
that ultimately abstract away the nature of wireless communication.
The inefficiency of graph-based protocols in the SINR model is
well documented and has been shown theoretically as well as
experimentally~\cite{MaheshwariJD08,Moscibroda2006Protocol}.

In seminal work, Moscibroda and Wattenhofer \cite{MoWa06}
propose the study of the \emph{scheduling complexity} of arbitrary set
of wireless links. 
Early work on approximation algorithms
\cite{moscibroda06b,MoscibrodaOW07,chafekar07} produced approximation
factors that grew with structural properties of the network.

The first constant factor approximation algorithm was obtained for
capacity problem for uniform power in \cite{GHWW09} (see also
\cite{HW09}) in $\mathbf{R^2}$ with $\alpha > 2$.
%
Fangh\"anel, Kesselheim and V\"ocking \cite{FKV09} gave an algorithm
that uses at most $O(OPT + \log^2 n)$ slots for the scheduling problem
with \emph{linear} power assignment $P_v = \ell_v^\alpha$,
that holds in general distance metrics.
Recently, Kesselheim and V\"ocking \cite{KV10} gave a distributed
algorithm that gives a $O(\log^2 n)$-approximation for the scheduling
problem with any fixed length-monotone and sub-linear power assignment.
Our capacity algorithm yields an improved $O(\log n)$-approximate
centralized algorithm.

The situation with power control has been more difficult.
Moscibroda, Oswald and Wattenhofer \cite{MoscibrodaOW07} showed that
using either uniform or linear power assignment results in only a 
$\Omega(n)$-approximation of the capacity and scheduling problems.
In terms of $\Delta$, their construction gives a $\Omega(\log
\Delta)$-lower bound. 
Avin, Lotker and Pignolet \cite{ALP09} show non-constructively that 
the ratio between the optimal capacity with power control and optimal
capacity with uniform power is $O(\log \Delta)$ in $\mathbf{R^1}$, for any $\alpha > 0$.
Fangh\"anel et al.~\cite{FKRV09} showed that any oblivious power
scheduling results in a $\Omega(n)$-lower bound. In terms of $\Delta$, this
was shown to be $\Omega(\log\log \Delta)$ in \cite{us:esa09full}.
A $O(\log\log \Delta \cdot \log n)$-approximation algorithm was given
in \cite{us:esa09full}, that holds in any \emph{doubling} metric where
$\alpha$ is strictly greater than the doubling coefficient
(generalizing the Euclidean setup).
This was shown to lead to $O(\log\log \Delta \cdot \log^3
n)$-approximation in general distance metrics \cite{FV09}.

Much of the difficulty of the power control scenario has to do with
the asymmetry of the communication links. Fangh\"anel et al.~\cite{FKRV09} 
introduced a bi-directional version, where both endpoints
of each link are senders and receivers (using the same power).
They give an algorithm for arbitrary metric 
that approximates scheduling within a $O(\log^{3.5+\alpha} n)$ factor.
This was improved to a $O(\log n)$-factor for doubling metrics in
\cite{us:esa09full}. 

In a very recent breakthrough, Kesselheim obtained 
a $O(1)$-approximation algorithm
for the capacity problem with power control for doubling metrics
and $O(\log n)$ for general metrics \cite{KesselheimSoda11}. 
In contrast, in the power control part of our
paper, we are interested in what can achieved through locally
obtainable oblivious assignments, whereas the power assignment in
\cite{KesselheimSoda11} is necessarily non-oblivious.


\section{Model and Preliminaries.}
\label{sec:model}

Given is a set $L = \{\ell_1, \ell_2, \ldots, \ell_n\}$ of links, where
each link $\ell_v$ represents a communication request from a sender
$s_v$ to a receiver $r_v$. 
The distance between two points $x$ and $y$ is denoted $d(x,y)$.
The asymmetric distance from link $\ell_v$ to link $\ell_w$ is the distance from
$v$'s sender to $w$'s receiver, denoted $d_{vw} = d(s_v, r_w)$.
The length of link $\ell_v$ is denoted 
simply $\ell_v$.
%
Let $\Delta$ denote the ratio between the maximum and minimum length of a link.
We assume that each link has a unit-traffic demand, and model the case
of non-unit traffic demands by replicating the links.

Let $P_v$ denote the power assigned to link $\ell_v$, or, in other words, $s_v$ transmits with power $P_v$. 
%
We assume the \emph{path loss radio propagation} model for the
reception of signals, where the signal received at point
$y$ from a sender at point $x$ with power $\mathcal{P}$  is $\mathcal{P}/d(x, y)^\alpha$ where the constant 
$\alpha > 0$ denotes the
path loss exponent. 
%
We adopt the \emph{physical model} (or \emph{SINR model})
of interference, in which a node $r_v$
successfully receives a message from a sender $s_v$ if and only if the
following condition holds:
\begin{equation}
 \frac{P_v/\ell_v^\alpha}{\sum_{\ell_w \in S \setminus  \{\ell_v\}}
   P_w/d_{wv}^\alpha + N} \ge \beta, 
 \label{eq:sinr}
\end{equation}
where $N$ is a universal constant denoting the ambient noise, $\beta \ge 1$ denotes the minimum
SINR (signal-to-interference-noise-ratio) required for a message to be successfully received,
and $S$ is the set of concurrently scheduled links in the same \emph{slot}.
%
We say that $S$ is \emph{SINR-feasible} (or simply \emph{feasible}) if (\ref{eq:sinr}) is
satisfied for each link in $S$. 

Let \prob{$\calP$-Capacity} denote the
problem of finding a maximum SINR-feasible set of links under power
assignment $\calP$.
Let \prob{PC-Capacity} denote the problem of finding a maximum
cardinality set of links and a power assignment that makes these links feasible.
Let $\PCopt(L)$ denote the optimal capacity of a linkset $L$ under
any power assignment. 

The above defines the physical model for uni-directional links.
In the bi-directional setting, 
the asymmetry between senders and receivers disappear.
The SINR constraint (\ref{eq:sinr}) changes only in that the definition
of distance between link changes to 
\[ d_{wv} = \min(d(s_w,r_v),d(s_v,r_w), d(s_w,s_v), d(r_w,r_v))\ . \]
With this new definition of inter-link distances, all other definitions and conditions remain unchanged.
For all problems we will consider, when the problem name is prefixed 
with \prob{Bi-}, we assume bi-directional communication.

A power assignment $P$ is \emph{length-monotone} if $P_v \ge P_w$ whenever
$\ell_v \ge \ell_w$ and \emph{sub-linear} if $\frac{P_v}{\ell_v^{\alpha}} \le \frac{P_w}{\ell_w^{\alpha}}$ whenever $\ell_v \ge \ell_w$.

There are two specific power assignments we will be using throughout the
paper.  The first is $\calU$, 
or the uniform power assignment, where
every link transmits with the same power. The second is $\calM$,
or
the mean power (or ``square-root'') assignment, where $P_v$ is proportional to
$\sqrt{\ell_v^\alpha}$. It is easy to verify that both $\calU$ and $\calM$
are length-monotone and sub-linear.

\myparagraph{Affectance.}
We will use the notion of \emph{affectance}, introduced in
\cite{GHWW09} and refined in \cite{HW09}, which has a number of
technical advantages.  
The affectance $a^P_w(v)$ of link $\ell_v$ caused by another link $\ell_w$,
with a given power assignment $P$,
is the interference of $\ell_w$ on $\ell_v$ relative to the power
received, or
  \[ a^P_{w}(v) 
     = c_v \frac{P_w/d_{wv}^\alpha}{P_v/\ell_v^\alpha}
     = c_v \frac{P_w}{P_v} \cdot \left(\frac{\ell_v}{d_{wv}}\right)^\alpha,
  \] 
where $c_v = \beta/(1 - \beta N \ell_v^\alpha/P_v)$ is a constant
depending only on the length and power of the link $\ell_v$. We will
drop $P$ when it is clear from context. 
Let $a_v(v) = 0$.
For a set $S$ of links and a link $\ell_v$, 
let $a^P_v(S) = \sum_{w \in S} a^P_v(w)$ and $a^P_S(v) = \sum_{w \in S} a^P_w(v)$.
Using this notation, Eqn.~\ref{eq:sinr} can be rewritten as $a^P_S(v)
\leq 1$, and this is the form we will use. This transforms the
relatively complicated Eqn.~\ref{eq:sinr} into an inequality involving
a simple sum that can be manipulated more easily.

We will frequently consider the affectances arranged as a matrix
$\affm = (a_{uv})$, where $a_{uv} = a_v(u)$. 
Note that feasibility of a set $S$ is equivalent to $\sum_{v} a_{uv}
\leq 1; \forall u \in S$, where $\affm = (a_{uv})$ 
is the appropriate affectance matrix. 

We will require two simple matrix-algebraic results. Before we state
them, we need a few definitions. For a matrix $\mathbf{M}$ (vector
$\mathbf{v}$), let $\trace(\mathbf{M})$ ($\trace(\mathbf{v})$) denote
the sum of entries in the matrix (vector).  We say that a matrix
$\mathbf{M} = (m_{ij})$ is \emph{$q$-approximately-symmetric} if, for
any two indices $i$ and $j$, $m_{ij} \le q \cdot m_{ji}$ for some $q
\geq 0$. For a matrix $\mathbf{M}$ and vectors $\mathbf{u}$ and
$\mathbf{v}$, we will use the notation $\mathbf{M} \cdot \mathbf{u}
\leq \mathbf{v}$ to mean $(\mathbf{M \cdot u})_i \leq \mathbf{v}_i$
for all $i$.


The proof of the following two claims can be found in the appendix (Section \ref{app:matrix}).
\begin{claim}
\label{matrixsumclaim}
Let $\mathbf{M}$ be a $n \times n$ matrix that contains only non-negative
entries. Assume $\trace(\mathbf{M}) \leq \gamma n$ for some $\gamma \geq 0$. Then, for any $\lambda > 1$, there are at least 
$(1 - \frac{1}{\lambda}) n$ rows (columns) $\mathbf{r}$ of $\mathbf{M}$ for which $\trace{(\mathbf{r})} \leq \gamma \cdot \lambda$.
\end{claim}

\begin{lemma}
Let $n$ be a number and $\mathbf{p}$ be a positive $n$-dimensional vector.
Let $\mathbf{M}$ be a non-negative $q$-approximately-symmetric $n$-by-$n$
matrix such that $\mathbf{M} \mathbf{p} \le \mathbf{p}$.
Then, $\trace(\mathbf{M}) \le (q+1) n$.
\label{lem:nearly-symm}
\end{lemma}

\myparagraph{Signal-strength and robustness.}
A \emph{$\delta$-signal} set is one where the
affectance of any link is at most $1/\delta$.
A set is SINR-feasible iff it is a 1-signal set.
Let $OPT_\delta^P = OPT_\delta^P(L)$ be the maximum number of links in a
$\delta$-signal set given a power assignment $P$ (and let $OPT_\delta^{P,B} = OPT_\delta^{P,B}(L)$ be the 
corresponding maximum for bi-directional links). We drop $\delta$, $B$ and
$P$ when understood from context.

The following results are essential.

\begin{proposition}[\cite{FKRV09},Prop.~3]
Let $L$ be set of links and $P$ be a power assignment
for $L$.
Then, for any $\tau \le \delta$, $OPT^P_\delta(L) \ge \frac{OPT^P_\tau(L)}{4 \delta/\tau} \quad \textrm{and}\quad
OPT_\delta^{P,B}(L) \ge \frac{OPT_\tau^{P,B}(L)}{8\delta/\tau}$. 
This is constructive.
\label{prop:signal-length}
\end{proposition}

\begin{lemma}[\cite{us:esa09full}]
Let $\ell_u, \ell_v$ be links in a $q^\alpha$-signal set.
Then, $d_{uv} \cdot d_{vu} \ge q^2 \cdot \ell_u \ell_v$. 
\label{lem:ind-separation}
\end{lemma}



\section{Fixed Power in General Metrics.}

In this section, we will deal with the \prob{$\calP$-Capacity} problem
for \emph{fixed} power assignments.  The most widely studied fixed power
assignments are \emph{oblivious}, 
where the power assigned to a link depends only on the length
of the link. Our results are not restricted to oblivious power, but we will
see that they do apply for popular oblivious assignments considered in the 
literature, such as $\calU, \calM$ and linear power (where $P_v = \ell_v^{\alpha}$).

We use the simple greedy algorithm \alg{C} described in Fig.~\ref{alg1fig}
for \prob{$\calP$-Capacity}.
Algorithm \alg{C}
is closely related to
Algorithm \alg{A}(c) in \cite{HW09}, differing primarily in that
it considers also the outgoing affectances from the
given link $\ell_v$. Also, the last line was added to ensure
feasibility in general metrics; as shown in \cite{HW09}, it is not
needed in $\mathbf{R^2}$ when $\alpha > 2$.

\begin{figure}[htbp]
\begin{pseudo}
  \textrm{Algorithm} \alg{C}(L,P)
     Suppose the links $L{}={}\{\ell_1,\ell_2,\ldots,\ell_n\}$ 
        are in non-decreasing order of length.  
     Let $\gamma = 1/2$
     $S\leftarrow\emptyset$
     for $v\leftarrow{}1$ to $n$ do
        if $a^P_S(v)+a^P_v(S)<\gamma$
           then add $\ell_v$ to $S$
     Output $X=\{\ell_v\in{}S:a_{S}^P(v)\le{}1\}$
\end{pseudo}
\caption{Algorithm for the capacity problem}
\label{alg1fig}
\end{figure}

\subsection{Constant Factor Approximation for \prob{$\calP$-Capacity}.}
We will show the following.
\begin{theorem}
\label{consfactorthm}
Algorithm \alg{C} approximates \prob{$\calP$-Capacity} within a constant factor for any
length-monotone, sub-linear $\calP$.
\end{theorem}

The feasibility of the solution output by \alg{C} is evident by virtue of the last
line of the algorithm.
We proceed to show that it gives constant-factor approximation.
We first compare the final output to the set $S$.

\begin{lemma}
Let $S$ and $X$ be the sets created by the algorithm $\alg{C}$ on a
given link set. Then, $|X| \ge |S| (1 - \gamma)$.
\label{lem:bound-on-x}
\end{lemma}

\begin{proof}
Observe that by the algorithm construction, the sum 
$\trace(\affm) = \sum_{v,w \in S} a_v(w) = \sum_{v \in S} a_S(v)$ is upper bounded by $\gamma |S|$.
Then the Lemma follows from Claim \ref{matrixsumclaim} by plugging in $\lambda = \frac{1}{\gamma}$.
\end{proof}

The main task is to relate the size of the set $S$ to the optimal solution.
For this, we extend the ``blue dominant centers'' lemma of
\cite{GHWW09} (and ``blue-shadowed lemma'' of \cite{HW09}) to
arbitrary metric spaces and length-monotone sub-linear power assignments.  

For a set $S$ and link $\ell_v$, let $S_v^- = \{\ell_w \in S :
\ell_w \leq \ell_v \}$ and $S_v^+ = \{\ell_w \in S : \ell_w \geq \ell_v \}$. 

\begin{lemma}[Red-Blue Lemma]
\label{blueshadowed}
Let $R$ and $B$ be disjoint sets of links, 
referred to as red and blue links in a length-monotone, sub-linear power metric.
If $|B| > 4 |R|$ and $B$ is a $3^\alpha$-signal set, 
then there is a blue link $\ell_b$ such that
$a_{R^-_b}(b) + a_b{(R^-_b)} \le 3^{\alpha} (a_{B}(b) + a_b(B))$.
\label{lem:rb-links}
\end{lemma}

This lemma is a consequence of the following two symmetric lemmas.

\begin{lemma}
\label{rblem1}
Let $R$ and $B$ be disjoint sets of links, 
referred to as red and blue links in a length-monotone power metric.
If $|B| > 2 |R|$ and $B$ is a $3^\alpha$-signal set, 
then there is a set $B' \subseteq B$ of size at least $|B| - 2|R|$
such that 
for all $\ell_b \in B'$,
$a_{R^-_b}(b) \le 3^{\alpha} a_{B}(b)$.
\end{lemma}

\begin{lemma}
\label{rblem2}
Let $R$ and $B$ be disjoint sets of links, 
referred to as red and blue links in a sub-linear power metric.
If $|B| > 2 |R|$ and $B$ is a $3^\alpha$-signal set, 
then there is a set $B' \subseteq B$ of size at least $|B| - 2|R|$
such that 
for all $\ell_b \in B'$,
$a_{b}(R^-_b) \le 3^{\alpha} a_{b}(B)$.
\end{lemma}

We prove Lemma \ref{rblem1} below and relegate proof of the very similar
 Lemma \ref{rblem2} to the appendix. 

\begin{proof}[Proof of Lemma \ref{rblem1}]
For each link $\ell_r \in R$, we will assign a set of ``guards" $X_r \subseteq B$.
For different $\ell_r$, the $X_r$'s will be disjoint and each $X_r$ will be of size at most $2$. 
We claim that each link $\ell_b$ that
remains in $B$ after removing all the guard sets satisfies the lemma.

Here's how we choose the guards. 
We process the links in $R$ in an arbitrary order.
Initially set $B' \leftarrow B$.
For each link $\ell_r \in R$, add to $X_r$ the link $\ell_x \in B^+_r$
with the sender nearest to
$s_r$ (if one exists); also add the link $\ell_y \in B^+_r$ with the
receiver nearest to $s_r$ (if one exists; possibly, $\ell_y = \ell_x$); 
finally, remove the links in $X_r$ from $B'$ and repeat the loop.

\begin{figure}[htbp]
  \begin{center}
    \includegraphics[height=5cm]{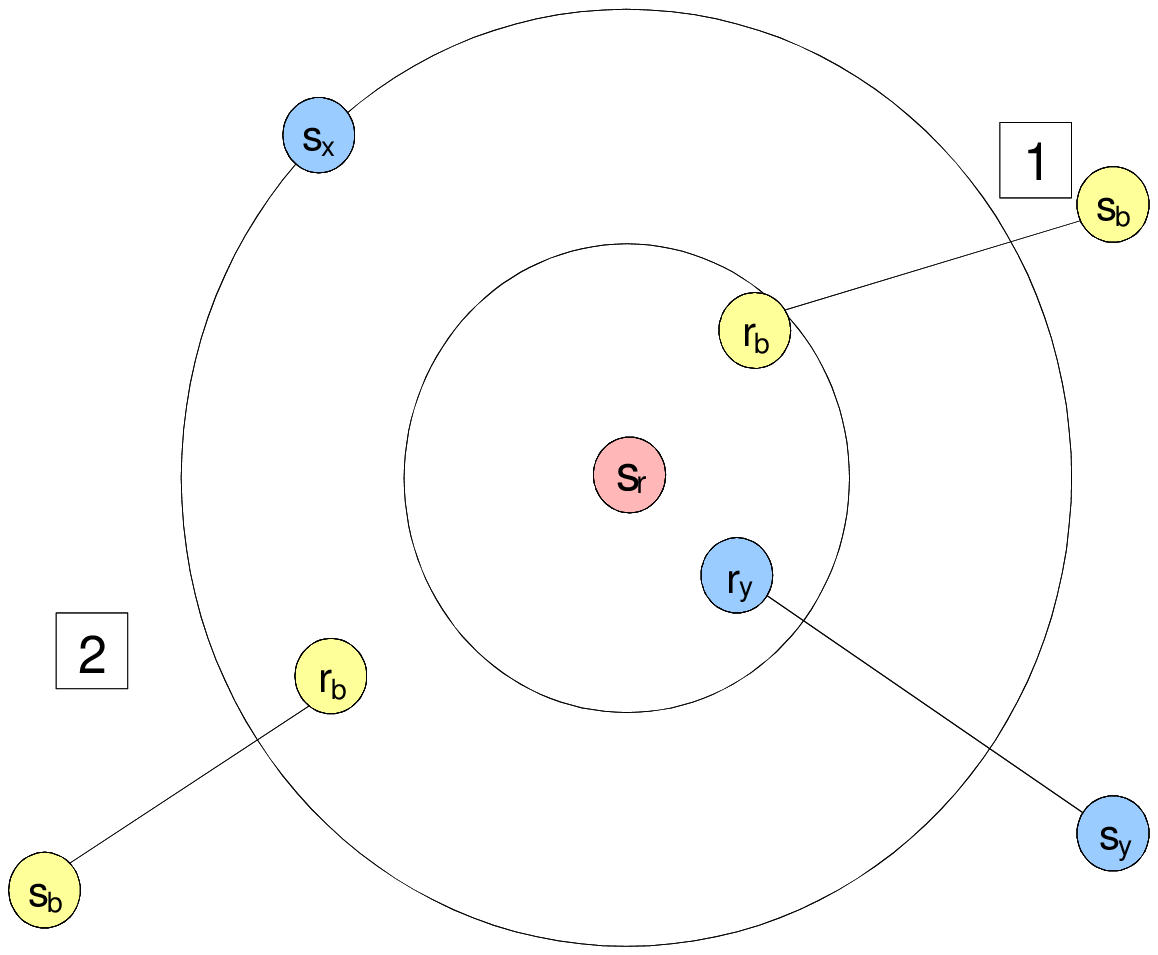}
  \end{center}
\caption{Example configuration of guards and points. Guard $s_x$
  (on left) is the nearest blue sender to red point $s_r$
  (in center) among those of links larger than $\ell_r$. }
\label{fig:red-blue}
\end{figure}

Since $|B| > 2 |R|$, $|B'| \geq |B| - 2|R| > 0$, by construction.
Consider any link $\ell_b \in B'$.
In what follows, we will show that the affectance of any link  
$l_r \in R^-_b$ on $\ell_b$ is comparable to the affectance of one of the guards of $l_r$
on $\ell_b$. Once we recall that the guards are blue, this implies that the overall affectance
on $\ell_b$ from $R^-_b$ is not much worse than that from $B$.

Since $\ell_b \not \in X_r$, $X_r$ is non-empty and contains guards
$\ell_x$ and $\ell_y$ (potentially $\ell_x = l_y$).
Consider any $\ell_r \in R^-_b$.
Let $d$ denote $d(s_{x}, s_r)$ and observe that
since $\ell_b \not \in X_r$, $d(s_b, s_r) \ge d$.

We claim that $d_{rb} \ge d/2$. Before proving the claim, let us see how 
this leads to the proof of the Lemma. If the claim is true, then by the triangular inequality,
$d_{x b} = d(s_{x}, r_b) \le d(s_x, s_r) + d(s_r, r_b) = d +
d_{rb} \le 3 d_{rb}\ .$
Since $\ell_x \geq \ell_r$ and the
power metric is length-monotone, $P_x \geq P_r$.
Thus, 
$ \frac{a_x(b)}{a_r(b)} = \frac{P_x}{P_r}
\left(\frac{d_{rb}}{d_{xb}}\right)^\alpha \ge 3^{-\alpha}\ . $
Summing over all links in $B$,
$ a_{B}(b) \ge \sum_{\ell_r \in R^-_b} a_{X_r}(b) \ge 3^{-\alpha} \sum_{\ell_r \in
  R^-_b} a_r(b) = 3^{-\alpha} a_{R^-_b}(b)\ . $

To prove the claim that $d_{rb} \ge d/2$, let us suppose otherwise for contradiction.
Then, by the triangular inequality,
$ \ell_b = d(s_b,r_b) \ge d(s_b,s_r) - d(r_b,s_r) > d - d/2 = d/2\ . $
Since $\ell_y$ was chosen into $X_r$, its receiver is at least as close to
$s_r$ as $r_b$, that is $d(r_y,s_r) \le d(r_b, s_r) < d/2$,
and its sender is also at least far as $s_{x}$, or
$d(s_y, s_r) \ge d$. So, $\ell_y \ge d(s_y,s_r) - d(s_r,r_y) > d/2$.
Now,
$d(r_{y}, r_b) \le d(r_y, s_r) + d(s_r, r_b) < d$ and
$ d_{y b} \cdot d_{b y} \le (\ell_y + d(r_y,r_b)) \cdot (\ell_b + d(r_y,r_b))
    < (\ell_{y}+d) \cdot (\ell_{b} + d) < 9 \cdot \ell_{y} \ell_b. $
But since $B$ is a $3^\alpha$-signal set,  
this is a contradiction by Lemma \ref{lem:ind-separation}.
Hence, any link $\ell_r$ in $R_b^-$ satisfies $d(s_r, r_b) = d_{rb} \ge d/2$. 
\end{proof}

\begin{lemma}
Let $S$ be the set found by the algorithm and $\tau = 3^{\alpha+1}/2\gamma$.
Then, $|S| \ge OPT_{2 \tau}/10$.
\label{lem:rb-applic}
\end{lemma}
\begin{proof}
By Claim~\ref{matrixsumclaim}, there is a set $O \subseteq OPT_{2 \tau}$ of size at least $OPT_{2 \tau}/2$ such that for all $u \in O$, $a_{u}(O) \leq \frac{1}{\tau}$. By definition, $a_O(u) \le \frac{1}{2\tau}$.

We claim that $|S| \ge |O|/5$.
Suppose otherwise. Then, $|S| < |O\setminus S|/4$.
Applying Lemma \ref{lem:rb-links} with $B =
O \setminus S$ and $R=S$, we find that there is a link $\ell_b$ in $O
\setminus S$ that satisfies 
$a_{S_b^-}(b) + a_b{(S_b^-)} \le 3^\alpha (a_{B}(b) + a_{b}(B))
\le 3^{\alpha}(\frac{1}{\tau} + \frac{1}{2\tau}) = \gamma$. 
The operation of the algorithm is then such that the algorithm would
have added the link $\ell_b$ to $S$, which is a contradiction.
\end{proof}

The proof of Thm.~\ref{consfactorthm} is now straightforward.
\begin{proof}
We can bound $ |X| \ge \frac{|S|}{2} \ge \frac{OPT_{2\tau}}{20} = \Omega(OPT)$. 
The first inequality is by Lemma \ref{lem:bound-on-x}, the second by 
Lemma \ref{lem:rb-applic}, and the last equality is a consequence of Proposition
\ref{prop:signal-length}.
\end{proof}


We note that our algorithm applies equally in the {\bf bi-directional} setting.

\subsection{The Necessity of Exponential Dependence on $\alpha$.}
In all results achieved in prior work as well as this work, the
approximation factor depends exponentially on $\alpha$. Here we show a
simple approximation preserving reduction from the maximum independent
set problem to prove that this dependence is necessary in
general. However, this does not rule out better dependence on $\alpha$
for special cases, say for small values of $\alpha$ on the
plane.

\begin{theorem}
\label{pathlossnecessary}
For $\alpha \ge \lg n + 1$ and any fixed $\epsilon > 0$, there can be no polynomial time algorithm that approximates 
$\prob{$\calU$-Capacity}$ in general metrics to within a factor better than
 $\Omega(2^{(1-\epsilon) \alpha})$, unless $P =NP$. 
\end{theorem}
\begin{proof}
By reduction from the maximum independent set problem in graphs.
Given graph $G(V, E)$, form a link $\ell_v$ for each vertex $v\in V$,
with $\ell_v = 1$.
If $\{u, v\} \in E$, set $d(s_v, r_u) = d(s_u,
r_v) = 1$, while if $\{u, v\} \not \in E$, set $d(s_v, r_u) = d(s_u, r_v) =
2$. It is easy to show that the triangle inequality holds.  Set
assume $N = 0$ and $\beta = 1.5$.  Any
set of vertices $S \subseteq V$ is an independent set if and only if
the corresponding set of links are SINR-feasible. To see this, first consider
a set of vertices that contains at least one edge $\{u, v\}$. Then
$a_S(v) \geq a_u(v) = \beta (\frac{1}{1})^{\alpha} > 1$. So, $S$ is not feasible.
On the other hand, if $S$ is an independent set, then for all $v \in S$,
\[ a_S(v) = \sum_{u \in S} a_u(v) = |S|\beta \left(\frac{1}{2}\right)^{\alpha}
   \le \frac{\beta}{2} \cdot |S| \left(\frac{1}{2}\right)^{\lg n}
   < 1\ . \]
The reduction is clearly approximation preserving.  Since independent
set is $\Omega(n^{1 - \epsilon})$-inapproximable unless P=NP (see
\cite{Zuckerman06}) for any fixed $\epsilon > 0$, the capacity problem
is hard to approximate within a factor of $\Omega(n^{1 - \epsilon}) =
\Omega(2^{(1-\epsilon) \alpha})$.
\end{proof}

\section{Approximating Power Control Using Oblivious Power.}

We will prove the following results.
\begin{theorem}
\label{pcunidir}
There is a $O(\log \log \Delta + \log n)$-approximation algorithm for
\prob{PC-Capacity} in general metrics that uses the mean power
assignment $\calM$.
\end{theorem}

\begin{theorem}
\label{pcbidir}
There is a $O(1)$-approximation algorithm for \prob{Bi-PC-Capacity} in
general metrics  that uses the mean power
assignment $\calM$.
\end{theorem}

Thus for bi-directional links,
we find that mean power is essentially the best possible. For uni-directional links, oblivious
power is known to be sub-optimal, but we achieve nearly tight bounds in light
of lower bounds in \cite{us:esa09full}.

For simplicity we assume $\beta=1$.
We shall also assume that the noise $N$ is negligible, 
which is acceptable in the power control case since
we can scale the powers to make the noise arbitrarily small. The following simple lemma 
makes this rigorous.

For any oblivious power assignment $P$ and any number $\mathbf{s} > 0$ define the assignment
$P^\mathbf{s}$ as $P_v^\mathbf{s} = \mathbf{s} P_v$. That is, $P^\mathbf{s}$ is a simple
linear scaling of $P$.

\begin{lemma}
Assume a given link set $L$ and ambient noise $N$.
Then for any oblivious power assignment $P$, there exists a number $\mathbf{s}$ such that the assignment
$P^\mathbf{s}$ has the following property: If a set $S \in L$ is feasible using $P^\mathbf{s}$ assuming
zero noise, then there is a set $S'$ of size $\Omega(S)$ which is feasible assuming ambient noise $N$.
Also, such an $\mathbf{s}$ can be found efficiently.
\end{lemma}
\begin{proof}
We will be comparing the case where the noise is zero versus the case
where the noise is $N$. We will use $c_v^0$, $a^0_S(v)$, $a^0_u(V)$
etc.\ to denote the value of $c_v$ and affectances in the zero noise
case, and the un-superscripted versions ($c_v, a_S(v), a_u(v)$ etc.)
to denote the noisy case.

First we claim that we can choose $\mathbf{s}$ such that
$c_v \leq 2$ for all $\ell_v$. Recall that $c_v = \beta / (1 - \beta N \ell_v^{\alpha} / P_v)$
depends only on $P_v$. Thus we can easily choose $\mathbf{s}$ such that 
 such that $N \ell_v^{\alpha} /P^{\mathbf{s}}_v \leq \frac{2}{\beta}$
for all $\ell_v$. Thus $c_v \leq 2 \beta \leq 2$ (since we have assumed $\beta = 1$).

On the other hand, if the noise level is zero then $c^0_v = \beta = 1$. 
Consider a set $S$ that is feasible under $P^{\mathbf{s}}$ assuming zero noise. 
Then $a^0_S(v) \leq 1$ for all $\ell_v$. 
Now, 
\begin{align*}
a_S(v) & = \sum_{u \in S}  a_u(v)  = c_v \sum_{u \in S} \frac{P^s_u}{P^s_v} \left(\frac{\ell_v}{d_{uv}}\right)^{\alpha} \\
  & \leq 2 c^0_v\sum_{u \in S} \frac{P^s_u}{P^s_v} \left(\frac{\ell_v}{d_{uv}}\right)^{\alpha}  
    = 2 \sum_{u \in S}  a^0_u(v) = 2 a^0_S(v) \\
  & \leq 2 \ .
\end{align*}
The statement of the Lemma now follows from Prop.~\ref{prop:signal-length}.
\end{proof}



\subsection{Uni-directional Capacity with Power Control via Oblivious Power.}

In \cite{us:esa09full}, Halld\'orsson achieved a $O(\log n \cdot \log
\log \Delta)$ approximation factor for uni-directional capacity with
power control for \emph{doubling} metrics via $\calM$.  Using tools
developed in this paper, the same approximation can be achieved for
general metrics in a straightforward manner.  We additionally improve
the bound to $O(\log n + \log \log \Delta)$.

First we need an weaker result on a class of power assignments we call
Lipschitz power assignments.

\begin{defn}
A power assignment $\calP$ is \emph{Lipschitz} if
there is some constant $c$ such that $\calP(\ell_v) \le c \cdot
\calP(\ell_u)$, for any two links $\ell_u, \ell_v$ with 
$\ell_u/2 \le \ell_v \le 2 \ell_u$.
\end{defn}

We prove the following.
\begin{theorem}
\label{lipdelta}
Let $\calP$ be any Lipschitz power assignment.
Then, for any linkset $L$,
$\PCopt(L) = O(\log \Delta) OPT_\calP(L)$.
\end{theorem}

For Lipschitz power assignments, the above result is the best
possible, given the matching lower bound for both uniform and linear 
power \cite{MoscibrodaOW07, ALP09}.

Before proving Thm.~\ref{pcunidir}, we need a weaker result for Lipschitz power assignments.

\begin{theorem}
\label{lipscitz-nealy-eql}
A set of links $L$ is said to contain nearly equilength links if for all $u, v \in L$,
$\frac{\ell_v}{2} \leq \ell_u \leq 2 \ell_v$.
Let $\calP$ be any Lipschitz power assignment.
Then, for any linkset $L$ of nearly equilength links, 
$\PCopt(L) = O(OPT_\calP(L))$.
\end{theorem}

\begin{proof}
First note that for nearly equilength links, the affectance for any
Lipschitz assignment $\calP$ varies from that of uniform power $\calU$
by no more than a constant factor. Thus by the signal strengthening
property (Proposition \ref{prop:signal-length}) we safely assume $\calP = \calU$.

Let $S$ be an optimal $6^\alpha$-signal subset of $L$.
Let $P$ be a power assignment for which $S$ is a $6^\alpha$-signal set.
Let $a_v^\calU(u) = (\ell_u/d_{vu})^{\alpha}$ be the affectance of
link $\ell_v$ on $\ell_u$ under uniform power. Then, $a^P_{v}(u) =  a_v^\calU(u) P_v/P_u$
is the affectance under $P$.
Viewed as a matrix $\affm = (a_{uv})$ where $a_{uv}=a_v^\calU(u); \forall u,v \in S$, feasibility of $S$ using $P$ 
is equivalent to
$\affm \mathbf{p} \le \mathbf{p}$ where $\mathbf{p} = [P_1, P_2 \ldots]^T$.


We claim that $d_{uv} \le 2 d_{vu}$ for all $u, v$. To show this assume $P_u \leq P_v$
and we will prove the inequality in both directions. Since 
$\frac{P_v}{P_u}\left(\frac{\ell_u}{d_{vu}}\right)^{\alpha} \leq 6^{-\alpha}$, we get 
$d_{vu} \geq 6 \ell_u \geq 2(\ell_u + \ell_v)$. 
By the triangular inequality $d_{uv} \ge d_{vu} -
(\ell_u + \ell_v) \ge d_{vu}/2$ and 
$d_{vu} \ge \max(2(\ell_u+\ell_v),d_{uv}-(\ell_u+\ell_v)) \ge 2 d_{uv}/3$.

Then,
$a^{\calU}_{u}(v) = (\ell_v/d_{uv})^{\alpha} \le (2\ell_u/\frac{1}{2}d_{vu})^{\alpha} \le 4^{\alpha} a^{\calU}_{v}(u)$.
Hence, $\affm$ is $4^{\alpha}$-nearly symmetric.
By Lemma \ref{lem:nearly-symm},
the set $S'$ of links with affectance at most $2 \cdot
4^{\alpha}$ is of size at least $|S|/2$.
By signal strengthening (Prop.~\ref{prop:signal-length}), applied both
to $S'$ and $S$, there is a 1-signal subset $S''$ in $S'$ of size
at least $|S'|/(8 \cdot 4^{\alpha})$, while $S \ge \PCopt/(4 \cdot 6^\alpha)$.
\end{proof}

Thm.~\ref{lipscitz-nealy-eql} can be used to achieve Thm.~\ref{lipdelta} (a generalization of a result from \cite{us:esa09full} to arbitrary distance metrics)
in a straightforward manner.

\begin{proof}[Proof of Thm.~\ref{lipdelta}]
Given an arbitrary linkset, we can divide it into $\log \Delta$
nearly-equilength groups. Let these groups be $L_i$ for $i = 1 \ldots \log \Delta$.
Then by Thm.~\ref{lipscitz-nealy-eql}, $\PCopt(L_i) = O(OPT_\calP(L_i))$. Let, 
$r = \arg \max \PCopt(L_r)$. Now,
\begin{align*}
\PCopt(L) & \leq \log \Delta \cdot  \PCopt(L_r) \\ 
& = O(\log \Delta \cdot OPT_\calP(L_r)) = O(\log \Delta \cdot OPT_\calP(L))\ .
\end{align*}
\end{proof}

We now combine the use of mean power $\calM$ with Thm.~\ref{lipdelta}
to get an algorithm with better approximation
ratio in terms of the dependence on $\Delta$.
To to do this, we modify ideas from \cite{us:esa09full}.
A key idea in  \cite{us:esa09full} was to partition the input links into length based
classes. We adopt this idea but employ a more efficient partitioning of the input links
and combine it with Thm.~\ref{lipdelta} for our result.

For this we will need the following lemma, a slight variation of 
Lemma 4.2 in \cite{us:esa09full}.

\begin{lemma}[\cite{us:esa09full}]
\label{wellsep}
Let $Q$ be a set of links that are SINR-feasible using some power assignment, and let
$\ell_v$ be a link that is shorter than the links in $Q$ by a factor of at least $n^{2/\alpha}$. Suppose 
$\max{(a^{\calM}_v(w), a^{\calM}_w(v))} \geq \frac{1}{2n}; \forall \ell_w \in Q$. Then $|Q| = O(\log \log \Delta)$.
\end{lemma}

\myparagraph{Partitioning the input into levels.}
For an input set $L$, assume the shortest link has length $\ell_{\min}$. Define non-overlapping sets,
$L_k = \{u \in L: \ell_{\min} D^{k+1} > \ell_u \geq \ell_{\min} D^k\}$ for all $k \geq 1$, 
where $D = 8 n^{2/\alpha}$. 
Let us call these sets \emph{levels}. 
We partition $L$ into $M_o$ and $M_e$, where
$M_o = \cup_{k \text{ is odd}} L_k$ and $M_e = \cup_{k \text{ is even}} L_k$.
$M_o$ (and $M_e$) is a union of levels such that intra-level distances
vary by no more than $D$ while inter-level distances vary by at least
$D$. In what follows, we will present an $O(\log n + \log \log
\Delta)$-approximation algorithm that takes as input such a
set. Since, $L = M_o \cup M_e$, we can claim the same approximation
factor for any input set with additional multiplicative factor of $2$.

\myparagraph{Algorithm and analysis.} 
We will adopt the convention that the levels are sorted in increasing
order, thus if $r > s$, links in $L_r$ are longer than links in $L_s$.

\begin{figure}[htbp]
\begin{pseudo}
  \alg{PC}($M_x=\cup_{r\geq0}L_r$)
     $S\leftarrow\emptyset$
     for $r\leftarrow{}0\ldots$ do
       Divide $L_r{}={}\cup{}Q_p$ such that lengths 
          of links within $Q_p$ varies by a 
          factor of at most $2$
       for each $Q_p$ do
           $S_p{}\leftarrow{}\emptyset$
           for $\ell_v{}\in{}Q_p$ in increasing order do
                add $\ell_v$ to $S_p$ if both:
                     1. $a_{S_p}(v){}+{}a_v(S_p)<\gamma$, and
                     2. $S$ contains no link $\ell_w$ s.t. 
                         $\max{(a^{\calM}_v(w),{}a^{\calM}_w(v))}{}\geq{}\frac{1}{2n}$
       $q\leftarrow{}\arg_p\max{|S_p|}$
       Set $S\leftarrow{}S\cup{}S_q$
     Output $X=\{\ell_v\in{}S:a_{S}(v)\le{}1\}$
\end{pseudo}
  \caption{Algorithm for \prob{PC-capacity}}
\end{figure}

\begin{lemma}
Consider the set $S$ selected by $\alg{PC}(L)$. Then $S = \Omega(\frac{\PCopt}{\log n + \log \log \Delta})$.
\end{lemma}
\begin{proof}
We partition both $\PCopt$ and $S$ into levels: 
$\PCopt = \cup_r O_r$ where $O_r = \PCopt \cap L_r$ and  $S = \cup_r S_r$ where $S_r = S \cap L_r$.  Also let $o_r = |O_r|, s_r = |S_r|$. 
We aim to prove $S = \Omega(\frac{\PCopt}{\log n + \log \log \Delta})$ or equivalently 
$\sum_r s_r = \Omega(\frac{1}{\log n + \log \log \Delta}) \sum_r o_r$.

Now let us define $T_{qr} \subseteq O_r$ for $r > q$ as the set of links in $O_r$ that fail condition 2 
due to links selected from level $q < r$. In other words, 
\begin{align*}
T_{qr}  = \{ w \in O_r : & \max{(a^{\calM}_v(w), a^{\calM}_w(v))} \geq \frac{1}{2n}, \\ 
& \text{for some } v \in S \cap L_q \} \ .
\end{align*}

\noindent Let $t_{qr} = |T_{qr}|$. By Lemma \ref{wellsep},
\begin{equation}
\sum_{r > q} t_{qr} \leq c_1 s_q \log \log \Delta \ ,
\label{app:tqrbnd}
\end{equation}
for some constant $c_1$.

With this in mind let us investigate what \alg{PC} does for a single $L_r$. We find that 
\alg{PC} runs on $L_r$ the algorithm suggested by Thm.~\ref{lipdelta}. 
This is done by partitioning the links in to nearly equilength sets $Q_p$, running
algorithm $\alg{C}$ separately on each $Q_p$ and choosing the largest feasible set found.
The only difference is due to condition 2, in other words, the $O(\log \Delta)$ approximation factor achieved is not
in relation to $O_r$ but rather $O_r \setminus \sum_{q<r} T_{qr}$.
Since for any $L_r$, $\log \Delta = \theta(\log n)$ by construction, we can claim that 
\[ s_r \geq \frac{1}{c \log n} (o_r - \sum_{q  < r} t_{qr})\ , \]
 for some constant $c$.
Using Eqn.~\ref{app:tqrbnd}, this implies that
\begin{align*}
c \log n \sum_r s_r & \geq \sum_r o_r - \sum_r \sum_{q < r} t_{qr} \\
 & =  \sum_r o_r - \sum_q \sum_{r > q} t_{qr} \\
 & \geq  \sum_r o_r - \sum_q c_1 s_q \log \log \Delta\ ,
\end{align*}
which yields that $(c \log n + c_1 \log\log \Delta) S \ge \PCopt$, as desired.
\end{proof}

To complete the proof of Thm.~\ref{pcunidir}, all we need is to relate $X$ to $S$.
For $S$, the average affectance from links at the same level is
bounded by construction.
Additionally, also by the design of the algorithm, one can observe that the 
total affectance from links at other levels is at most $\frac{1}{2n} \times n = \frac{1}{2}$.
Now arguing that $X$ is within a constant factor of $S$ takes the same route
we have already seen via Lemma ~\ref{matrixsumclaim} and Proposition ~\ref{prop:signal-length}.

\subsection{Bi-Directional Capacity with Power Control.}

Let $\biopt_P(L)$ denote the optimal bi-directional capacity with
respect to power assignment $P$ and
$\bipcopt(L)$ the optimal bi-directional capacity with respect to any
power assignment.

We show that $\calM$ is optimal for bi-directional capacity modulo
constant factors.
\begin{theorem}
For any linkset $L$, $\biopt_\calM(L) = O(\bipcopt(L))$.
\label{thm:mean-vs-any-bidi}
\end{theorem}


\begin{proof}
Let $P$ be a power assignment with which the optimal capacity
$\bipcopt(L)$ is attained, and let $S$ be the feasible subset in $L$
of cardinality $\bipcopt(L) = \biopt_P(L)$.
Let $P'$ be a function defined by 
$P(v) = P'(v) \cdot \calM(\ell_v) = P'(v) \cdot \sqrt{\ell_v^{\alpha}}$.
Recall that in the bi-directional setting,
$d_{uv} = d_{vu}$, and 
affectance under $\calM$ is symmetric and given by
$a_v^\calM(w) = \left(\frac{\sqrt{\ell_v\ell_w}}{d_{vw}}\right)^{\alpha} = a^\calM_w(v)$.
Observe that 
\begin{align*}
 a_v^P(u) & = \frac{P(v)}{P(u)} \cdot
 \left(\frac{\ell_u}{d_{vu}}\right)^\alpha \\
    & = \frac{P'(v) \ell_v^{\alpha/2}}{P'(u) \ell_u^{\alpha/2}} \cdot 
          \left(\frac{\ell_u}{d_{vu}}\right)^\alpha \\
    & = \frac{P'(v)}{P'(u)} \cdot a_{v}^\calM(u) \ .
\end{align*}
Since $S$ is feasible w.r.t.\ $P$,
it holds that for each link $\ell_u$ in $S$, 
$\sum_{\ell_v \in S} a^P_v(u) \le 1$, or
$\sum_{\ell_v \in S} P'(v) a^\calM_{v}(u) \le P'(u)$.
Thus, if $\mathbf{A'} = (a'_{uv})$ where $a'_{uv} = a_{v}^\calM(u)$ is the affectance matrix for $\calM$,
we have that $\mathbf{A'} \mathbf{p'} \le \mathbf{p'}$, where 
$\mathbf{p'}$ is the vector of the values $P'(u)$ for all $u$.
Since $\mathbf{A'}$ is symmetric, 
$\trace(\mathbf{A'}) \le 2 n$, by Lemma \ref{lem:nearly-symm},
and by Claim~\ref{matrixsumclaim} there is a subset $S'$ of $S$ with at least $|S|/2$ links 
such that $\sum_{v \in S'} a'_{uv} \le 4$ for each link $\ell_u \in S'$.
Namely, for each link $\ell_u \in S'$,
\[ a_{S'}^\calM(u) = \sum_{v \in S'} a_{v}^\calM(u) = \sum_{v \in S'}
a'_{uv} \le 4 \ . \]
Hence, $S'$ is a $1/4$-signal set under $\calM$.
Finally, by signal-strengthening, there is a 1-signal subset $S''$ in
$S'$ of size $|S'|/32 \ge |S|/64  = \Omega(\bipcopt(L))$.
\end{proof}

In light of the fact that \alg{C(L)} is a constant factor algorithm for  $\calM$ by Thm.~\ref{consfactorthm} (which works for bi-directional links as well), Thm.~\ref{pcbidir}
is a simple consequence of Thm.~\ref{thm:mean-vs-any-bidi}
.

\section{Scheduling vs.\ Capacity.}

Our capacity results immediately imply results for the corresponding
scheduling problems, within a logarithmic factor. We state these
below. Recall that the scheduling problem is to the capacity problem
as what the graph coloring problem is to the independent set problem.

\begin{corollary}
There is an $O(\log n)$-approximation algorithm for link scheduling 
under any fixed length-monotone,
sub-linear power assignment $\calP$.
Further, for scheduling with power control, by using square-root
power, we get both $O(\log \Delta \log n)$- and 
$O((\log \log \Delta + \log n) \log n)$-approximation for
uni-directional case and $O(\log n)$-approximation for bi-directional case.
All results apply in general distance metrics and for all fixed $\alpha > 0$.
\end{corollary}

The relationship between capacity and scheduling has remained an open
question for some time. For instance, while a constant factor
approximation for capacity with uniform power (in \emph{doubling} metrics) was
obtained in \cite{GHWW09}, a corresponding result for scheduling has
been elusive, with a faulty claim in \cite{HW09}.
We show here that scheduling can indeed be more difficult by a
logarithmic factor. 

\subsection{A Lower Bound.}
In this section we provide bounds on how well oblivious power assignments can do in a general metric space. 
Lower bounds of this type have been proved previously for uni-directional links on the line \cite{FKRV09,us:esa09full},
giving
$\Omega(n)$ and $\Omega(\log \log \Delta)$ lower bounds. However, these
arguments don't say anything about bi-directional scheduling.

We  show a construction of equilength links that are feasible with
some power assignment, but  using oblivious power, one needs $\Omega(\log n)$ slots to schedule the links.
The reader may recall that for the capacity problem, in contrast, $O(1)$-approximation is possible for
the bi-directional case via oblivious power (Thm.~\ref{thm:mean-vs-any-bidi}).
Since the links have equal lengths, the only possible oblivious assignment is $\calU$.
The construction works for both uni-directional and bi-directional cases.

We start with a single link $\ell_{\hi}$ and define the other links in the link set $L$ in relation to $\ell_{\hi}$. For some 
suitably large $K$, define sets $S_k$ of links 
 for $k = 1 \ldots K$, such that $|S_k| = 2^{2(k-1)}$, and for all 
$i \in S_k$, $d_{i \hi} = (\gamma |S_k|)^\frac{1}{\alpha}$, where
 $\gamma$ is a fixed constant to
be defined later. Now, $L = \{\ell_{\hi}\} \cup (\cup_k S_k)$, and the total number of links $n = |L| = 1+ \sum_{k = 1 \ldots K} |S_k|$. It is easy to see that
\begin{equation}
\label{KeqLogn}
K = \theta(\log n)\ .
\end{equation}
Furthermore, let $\ell_i = 1, \forall i$ and $d_{ij} = d_{ji} = d_{i \hi} + d_{j \hi} , \forall i, j \neq \hi$.
All other distances are defined by transitivity. Let $N = 0$ and $\beta = 1$.

\begin{theorem}
\label{th:lb}
There exists a power-assignment $P$ such that $L$ is SINR-feasible
using $P$, whereas $L$ cannot be partitioned into $o(\log n)$
SINR-feasible sets using uniform power $\calU$. 
\end{theorem}

The following two technical lemmas imply the two claims in Thm.~\ref{th:lb}. Let
the $n \times n$ matrix $\affm = a_{ij}$ be defined by entries $a_{ii} = 0$ and $a_{ij} = a^\calU_j(i)$, and the vector
$\mathbf{p}$ be defined by $\mathbf{p}_i = P_i$.
\begin{lemma}
\label{lbmain1}
There exists a power assignment $P$ such that
$(\affm \mathbf{p})_i \leq \mathbf{p}_i, \forall i \in L \label{eqn2}$.
\end{lemma}

\begin{lemma}
\label{lbmain2}
Take any partition of $L$ into $r = o(\log n)$ sets $Q_i$ (
$i = 1 \ldots r$). Then there exists a link $j$ in some $Q_i$ such that $\sum_{\hj \in Q_i} a_{j \hj} > 1$, thus
proving that $Q_i$ is not feasible using $\calU$.
\end{lemma}

Recalling that SINR condition is equivalent to $a^P_L(i) = \frac{1}{\mathbf{p}_i}\sum_{j \neq i} a_{ij} \mathbf{p}_j \leq 1$, we see that
Lemma \ref{lbmain1} shows that the links are SINR-feasible with power assignment $P$. 

\begin{proof}[Proof of Lemma \ref{lbmain1}]
We will use the power assignment $P$ defined by, $P_{\hi} = 1$ and $P_i = \frac{1}{2^{k}}, \forall i \in S_k$.
Note that $a_{ij} = \left(\frac{\ell_i}{d_{ji}}\right)^{\alpha} = \frac{1}{d_{ji}^{\alpha}}$.



\noindent For link $\ell_\hi$, 
\begin{align*}
(\affm\mathbf{p})_{\hi} & = \sum_{i \in L} a_{\hi i} P_i = \sum_k \sum_{i \in S_k} P_i a_{\hi i} \\
 & = \sum_k \sum_{i \in S_k} \frac{1}{2^{k}} \frac{1}{\gamma |S_k|} = \sum_k \frac{1}{2^{k}} \frac{1}{\gamma}
 \leq \frac{1}{\gamma} \leq 1 = P_{\hi}\ ,
\end{align*}
for $\gamma \ge 1$. 
The first equality is simply the vector multiplication written out explicitly, the second equality
divides the sum into the sums for sets $S_k$. The third equality uses the explicit values of the power assignment and
$a_{\hi j} =  \frac{1}{\gamma |S_k|}$ for $j \in S_k$. The first inequality
follows from the fact that $\sum_{k\ge 1}{\frac{1}{2^k}} = 1$.

Now consider $j \in S_{\tilde{k}}$ for some $\tilde{k}$. We write
\begin{equation}
  \label{eqn:lb3terms}
(\affm\mathbf{p})_{j} = \sum_{k > \tilde{k}} \sum_{i \in S_k} P_i a_{j i}
   + \sum_{k \leq \tilde{k}} \sum_{i \in S_k} P_i a_{j i} + P_{\hi} a_{j \hi}
\end{equation}
and bound these three terms separately. 
The last term is $P_{\hi} a_{j \hi} = \frac{1}{\gamma |S_{\tilde{k}}|} = \frac{1}{ \gamma 2^{2(\tilde{k} - 1)}} \leq \frac{4}{\gamma 2^{\tilde{k}}} = \frac{4 P_j}{\gamma}$. 

Recall that if $k > \tilde{k}$, then $d_{i j} \geq d_{i \hi} \geq (\gamma |S_k|)^{1/\alpha}$ for $i \in S_k$,
while if $k \leq \tilde{k}$, then $d_{i j} \geq (\gamma |S_{\tilde{k}}|)^{1/\alpha}$ for
$i \in S_k$.
The first term in Eqn.~\ref{eqn:lb3terms} is bounded by 
\begin{align*}
\sum_{k > \tilde{k}} \sum_{i \in S_k} P_i a_{j i}
 & \leq \sum_{k > \tilde{k}} \sum_{i \in S_k} \frac{1}{2^{k}} \frac{1}{\gamma |S_k|} 
    =  \sum_{k > \tilde{k}} \frac{1}{2^{k}} \frac{1}{\gamma}  \\
& \le \frac{1}{\gamma} \frac{1}{2^{\tilde{k}}} \sum_{m \ge 1}  \frac{1}{2^m}
    = \frac{1}{\gamma} \frac{1}{2^{\tilde{k}}} = \frac{1}{\gamma} P_{j}\ .
\end{align*}
The first inequality is simply by writing out the explicit values of $P_i$ and the bound $a_{ji} = \frac{1}{d_{ij}^\alpha} \leq \frac{1}{\gamma |S_k|}$. 
Everything else is manipulation,
the final equality is a consequence of the definition of $P_j$.
Finally, the second term in Eqn.~\ref{eqn:lb3terms} is bounded by 
\begin{align*}
\lefteqn{\sum_{k \leq \tilde{k}} \sum_{i \in S_k} P_i a_{j i}} \\
 &  \leq \sum_{k \leq \tilde{k}} |S_k| \frac{1}{2^{k}} \frac{1}{\gamma |S_{\tilde{k}}|}  
  = \frac{1}{\gamma}  \sum_{k \leq \tilde{k}} 2^{2 (k - 1) - k - 2(\tilde{k}-1)} \\
 & = \frac{1}{\gamma 2^{\tilde{k}}}  \sum_{k \leq \tilde{k}} 2^{k - \tilde{k}}
 \le \frac{2}{\gamma 2^{\tilde{k}}}  = \frac{2}{\gamma} P_{j} \ .
\end{align*}
%
The first inequality follows from the definition of $P_i$ and $a_{ji}
= \frac{1}{d_{ij}^\alpha} \leq \frac{1}{\gamma |S_{\tilde{k}}|}$. The
first equality follows from the value of $|S_k|$. The rest are
consequences of basic facts of the geometric series, and the values of
$|S_{\tilde{k}}|$ and $P_j$. 


Setting $\gamma = 14$, $(\affm\mathbf{p})_{j} \leq \frac{4}{\gamma} P_{j} + \frac{2}{\gamma} P_{j} + \frac{1}{\gamma} P_{j} \leq \frac{1}{2} P_j$. This completes the proof of the Lemma.
\end{proof}

To prove Lemma \ref{lbmain2}, first we need a simple lemma about
partitions of sets. 

\begin{lemma}
\label{simplesetsum}
Consider a set $S = \cup S_{k}$ for $k = 1 \ldots K$ where the $S_k$ are disjoint. Consider any arbitrary partition of $S$ into $r$ parts labelled $Q_1 \ldots Q_r$ where $S = \cup_{i \leq r} Q_i$. Define $f_{ik} = \frac{|S_k \cap Q_i|}{|S_k|}$ to be the fraction of 
$S_k$ that falls in $Q_i$. Then there exists $i \leq r$ such that 
$\sum_{k} f_{ik} \geq \frac{K}{r}\ .$
\end{lemma}
\begin{proof}
From the definition of $f_{ik}$ we see that $\sum_{i \leq r} f_{ik} = 1$. Then
$ \sum_{i \leq r} \sum_{k} f_{ik} =  \sum_{k} \sum_{i \leq r} f_{ik} =  \sum_{k} 1 = K \  . $ The claim now follows.
\end{proof}

\begin{proof}[Proof of Lemma \ref{lbmain2}]
Let us invoke Lemma \ref{simplesetsum} and consider the set $Q_i$ for which the claim of the lemma holds. Thus,
$\sum_{k} f_{ik} \geq \frac{K}{r} = \Omega\left(\frac{\log n}{r}\right)$ (the last equality follows from Eqn.~ (\ref{KeqLogn})). 
Since the set $Q_i$ is clearly non-empty, there is some minimum value $\tilde{k}$ for which 
an element of $S_{\tilde{k}}$ exists in $Q_i$. Then, $\sum_{k > \tilde{k}} f_{ik} \geq \Omega\left(\frac{\log n}{r}\right) - 1$. 
Consider an arbitrary link $j \in S_{\tilde{k}} \cap Q_i$.

Note that for $\hj \in S_k$ where $k > \tilde{k}$, $d_{\hj j} = (\gamma |S_k|)^{1/\alpha} + (\gamma |S_{\tilde{k}}|)^{1/\alpha} \leq 2 (\gamma |S_k|)^{1/\alpha}$.
Now,
\begin{align*}
\sum_{\hj \in Q_i} a_{j \hj} & \geq \sum_{k > \tilde{k}} \sum_{\hj \in Q_i \cap S_k} a_{j \hj} \\
& \geq \frac{1}{\gamma 2^\alpha }\sum_{k > \tilde{k}} \sum_{\hj \in Q_i \cap S_k} \frac{1}{|S_k|} \\
& = \frac{1}{\gamma 2^\alpha }\sum_{k > \tilde{k}} f_{ik}|S_k| \frac{1}{|S_k|} \\
& = \frac{1}{\gamma 2^\alpha }\sum_{k >\tilde{k}} f_{ik} \geq \frac{1}{\gamma 2^\alpha } \left (\Omega\left(\frac{\log n}{r}\right) - 1\right) 
& > 1 \ ,
\end{align*}
for $r = o(\log n)$.
The first inequality comes from dividing the sum in to sums for different $S_k$'s. The second inequality comes from
the upper bound on $d_{\hj j}$. The first equality follows from the definition of $f_{ik}$. The third inequality follows
from the choice of $i$.

This proves that one cannot schedule the links in $o(\log n)$ slots using uniform (and thus, oblivious) power. 
\end{proof}


%
%

\bibliographystyle{abbrv}
\bibliography{./references}

\appendix

\section{Missing proofs}

\subsection{Matrix-algebraic Claims}
\label{app:matrix}
\begin{proof}[Proof of Claim \ref{matrixsumclaim}]
Suppose otherwise. Then, there are at least $\frac{n}{\lambda}$ rows $\mathbf{r}$ such that $\trace(\mathbf{r}) > \gamma \cdot \lambda$. Call the set of these rows $X$.
Then
\[\trace(\mathbf{M}) \geq \sum_{r \in X} \trace(\mathbf{r}) > \frac{n}{\lambda} \cdot \gamma \cdot \lambda 
= \gamma n \ ,\]
which is a contradiction. The proof for columns is nearly identical.
\end{proof}

\begin{proof}[Proof of Lemma \ref{lem:nearly-symm}]
The claim is proved by induction.
The case $n=1$ is immediate since $m_{11}\mathbf{p}_1 \le \mathbf{p}_1$, so $m_{11} \le
1=n \le (q+1)n$.

Suppose the claim holds for $n=k-1$, and consider the case $n=k$.
Assume without loss of generality that $\mathbf{p}_n$ is the smallest value in
$\mathbf{p}$; we otherwise rearrange the columns and rows accordingly.

By the matrix property,
$ \sum_i m_{ni} \mathbf{p}_i \leq \mathbf{p}_n\ . $
Since $\mathbf{p}_i \ge \mathbf{p}_n$ by assumption, 
we have that $ \sum_i m_{ni}  \leq \sum_i m_{ni} \frac{\mathbf{p}_i}{\mathbf{p}_n} \leq 1\ . $
%
%
By $q$-approximate symmetry,
$ \sum_{i} m_{in} \le q \sum_{i} m_{ni} \le q$.
It follows by the inductive hypothesis that
$ \sum_{ij} m_{ij} = \sum_{i=1}^{n-1} \sum_{j=1}^{n-1} m_{ij} 
         + \sum_{i=1}^n m_{in} + \sum_{i=1}^n m_{ni} - m_{nn}
   \le (q+1)(n-1) + q + 1 =  (q+1) n, $
establishing the claim.
\end{proof}

\subsection{Proof of Lemma \ref{rblem2}}

To prove this part of the result, we essentially switch the role of
senders and receivers, and then the proof follows that of Lemma \ref{rblem1}.

For each link $\ell_t \in R$, we will assign a set of ``guards" $X_t \in B$.
For different $t$, the $X_t$'s will be disjoint and each $X_t$ will be of size at most $2$. 
We claim that each link $\ell_b$ that
remains in $B$ after removing all the guard sets satisfies the lemma.

We process the links in $R$ in an arbitrary order.
Initially set $B' \leftarrow B$.
For each link $\ell_t \in R$, add to $X_t$ the link $\ell_y \in B^+_t$
with the sender nearest to
$r_t$ (if one exists); also add the link $\ell_z \in B^+_t$ with the
receiver nearest to $r_t$ (if one exists; possibly, $\ell_y = \ell_z$); 
finally, remove the links in $X_t$ from $B'$ and repeat the loop.

Since $|B| > 2 |R|$, $B' \geq |B| - 2|R| > 0$, by construction.
Consider any link $\ell_b \in B'$.
In what follows, we will show that the affectance of $\ell_b$ on any link  
$\ell_r \in R^-_b$ is comparable to the affectance of $\ell_b$ on one of the guards of $\ell_r$. 
Once we recall that the guards are blue, this implies that the overall affectance
on $R^-_b$ from $\ell_b$ is not much higher than that on $B$.

Consider any $\ell_r \in R^-_b$.
Since $\ell_b \not \in X_r$, $X_r$ contains two guards,
$\ell_y$ and $\ell_z$ (possibly the same).
Let $d$ denote $d(r_{z}, r_t)$ and observe that
since $\ell_b \not \in X_r$, $d(r_b, r_t) \ge d$.

We claim that $d_{bt} \ge d/2$. Before proving the claim, let us see how 
this leads to the proof of the Lemma. If the claim is true, then by the triangular inequality,
\begin{align}
\label{rbdistbound2}
d_{bz} = d(s_{b}, r_z) & \le d(s_b, r_t) + d(r_t, r_z) \\ 
& = d_{bt} + d \le 3 d_{bt} \ . \nonumber 
\end{align}

Thus, 
\begin{align*}
 \frac{a_b(z)}{a_b(t)} & = \frac{c_z}{c_t}\frac{P_b/d_{bz}^{\alpha}}{P_z/\ell_z^{\alpha}} \cdot \frac{P_t/\ell_t^{\alpha}}{P_b/d_{bt}^{\alpha}}  \\
& = \frac{P_t/\ell_t^{\alpha}}{P_z/\ell_z^{\alpha}} \cdot \left(\frac{d_{bt}}{d_{bz}}\right)^{\alpha} \frac{c_z}{c_t}
\geq \left(\frac{d_{bt}}{d_{bz}}\right)^{\alpha} \geq 3^{-\alpha}\ . 
\end{align*}
The first inequality follows from sub-linearity, and the second from Eqn.~\ref{rbdistbound2}.
Summing over all links in $B$,
\[ a_{b}(B) \ge \sum_{t \in R^-_b} a_b{(X_t)} \ge 3^{-\alpha} \sum_{\ell_t \in
  R^-_b} a_b(t) = 3^{-\alpha} a_b{(R^-_b)}\ . \]

To prove the claim that $d_{bt} \ge d/2$, let us suppose otherwise for contradiction.
Then, by the triangular inequality,
\[ \ell_b = d(s_b,r_b) \ge d(s_b,s_t) - d(s_t,r_b) > d - d/2 = d/2\ . \]
By the choice of $\ell_y$, its sender is at least as close to
$r_t$ as $b$, that is $d(s_y,r_t) \le d(s_b, r_t) < d/2$,
and its receiver is also at least far as $r_{z}$, or
$d(r_y, r_t) \ge d$. So, $\ell_y \ge d(r_y,r_t) - d(r_t,s_y) > d/2$.
But, that implies $\ell_{y}$ and $\ell_b$ are too close together, because
$d(s_{y},s_b) \le d(s_y, r_t) + d(s_b, r_t) < d$ and
\begin{align*} 
d_{y b} \cdot d_{b y} & \le (\ell_b + d(s_y,s_b)) \cdot (\ell_y + d(s_y,s_b)) \\
   & < (\ell_{b}+d) \cdot (\ell_{y} + d) < 9 \cdot \ell_{y} \ell_b. 
\end{align*}
Since $B$ is a $3^\alpha$-signal set,  
$b$ and $y$ cannot coexist in $B$ by Lemma \ref{lem:ind-separation},
which contradicts our assumption. 
Hence, any link $\ell_t \in R_b^-$ satisfies $d(s_b, r_t) = d_{bt} \ge d/2$. 

\end{document}